\documentclass[11pt]{article}
\usepackage{fullpage,comment,algorithm,algorithmic}
\usepackage{color,colortbl}
\usepackage{float}
\usepackage{amsthm}
\usepackage{amsmath}
\usepackage{amssymb}
\usepackage{graphicx}
\usepackage{xfrac}
\usepackage[colorinlistoftodos,prependcaption,textsize=tiny]{todonotes}
\usepackage{multirow}

\ifx\pdftexversion\undefined
\usepackage[colorlinks,linkcolor=black,filecolor=black,citecolor=black,urlco
lor=black,pdfstartview=FitH,pagebackref]{hyperref}
\else
\usepackage[colorlinks,linkcolor=blue,filecolor=blue,citecolor=blue,urlcolor
=blue,pdfstartview=FitH,pagebackref]{hyperref}
\fi

\newtheorem{theorem}{Theorem}
\newtheorem{lemma}{Lemma}
\newtheorem*{lemma*}{Lemma}
\newtheorem{corollary}[lemma]{Corollary}
\newtheorem{remark}{Remark}
\newtheorem{claim}[lemma]{Claim}

\newtheorem{definition}{Definition}

\newcommand{\af}[1]{\textcolor{red}{\textsf{#1}}}
\newcommand{\namedref}[2]{\hyperref[#2]{#1~\ref*{#2}}}
\newcommand{\sectionref}[1]{\namedref{Section}{#1}}

\newcommand{\theoremref}[1]{\namedref{Theorem}{#1}}
\newcommand{\defref}[1]{\namedref{Definition}{#1}}

\newcommand{\claimref}[1]{\namedref{Claim}{#1}}
\newcommand{\lemmaref}[1]{\namedref{Lemma}{#1}}
\newcommand{\tableref}[1]{\namedref{Table}{#1}}

\newcommand{\corollaryref}[1]{\namedref{Corollary}{#1}}

\newcommand{\algref}[1]{\namedref{Algorithm}{#1}}

\newcommand{\lineref}[1]{\namedref{Line}{#1}}

\newcommand{\ddim}{{\rm ddim}}
\newcommand{\supp}{{\rm supp}}

\newcommand{\diam}{{\rm diam}}

\newcommand{\poly}{{\rm poly}}
\newcommand{\polylog}{{\rm polylog}}

\newcommand{\argmax}{\rm argmax}

\newcommand{\eps}{\epsilon}
\newcommand{\Oeps}{O_\eps}

\def\inline#1:{\par\vskip 7pt\noindent{\bf #1:}\hskip 10pt}

\def\inline#1:{\par\vskip 7pt\noindent{\bf #1:}\hskip 10pt}

\def\blackslug{\hbox{\hskip 1pt \vrule width 4pt height 8pt
		depth 1.5pt \hskip 1pt}}

\def\QED{\quad\blackslug\lower 8.5pt\null\par}

\newcommand{\alert}[1]{\textbf{\color{red}
[[[#1]]]}\marginpar{\textbf{\color{red}**}}\typeout{ALERT:
\the\inputlineno: #1}}

\floatstyle{ruled}
\newfloat{algorithm}{tbp}{loa}
\providecommand{\algorithmname}{Algorithm}
\floatname{algorithm}{\protect\algorithmname}

\makeatother

\usepackage{authblk}
\usepackage{pdfsync}

\begin{document}
\author{Arnold Filtser}
\author{Ofer Neiman}
\affil{Ben-Gurion University of the Negev. Email: \texttt{\{arnoldf,neimano\}@cs.bgu.ac.il}}
\title{Light Spanners for High Dimensional Norms via Stochastic Decompositions}
\maketitle
\begin{abstract}
Spanners for low dimensional spaces (e.g. Euclidean space of constant dimension, or doubling metrics) are well understood. This lies in contrast to the situation in high dimensional spaces, where except for the work of Har-Peled, Indyk and Sidiropoulos (SODA 2013), who showed that any $n$-point Euclidean metric has an $O(t)$-spanner with $\tilde{O}(n^{1+1/t^2})$ edges, little is known.

In this paper we study several aspects of spanners in high dimensional normed spaces.
First, we build spanners for finite subsets of $\ell_p$ with $1<p\le 2$. Second, our construction yields a spanner which is both sparse and also {\em light}, i.e., its total weight is not much larger than that of the minimum spanning tree.
In particular, we show that any $n$-point subset of $\ell_p$ for $1<p\le 2$ has an $O(t)$-spanner with $n^{1+\tilde{O}(1/t^p)}$ edges and lightness $n^{\tilde{O}(1/t^p)}$.

In fact, our results are more general, and they apply to any metric space admitting a certain low diameter stochastic decomposition. It is known that arbitrary metric spaces have an $O(t)$-spanner with lightness $O(n^{1/t})$. We exhibit the following tradeoff: metrics with decomposability parameter $\nu=\nu(t)$ admit an $O(t)$-spanner with lightness $\tilde{O}(\nu^{1/t})$.
For example, $n$-point Euclidean metrics have $\nu\le n^{1/t}$, metrics with doubling constant $\lambda$ have $\nu\le\lambda$, and graphs of genus $g$ have $\nu\le g$. While these families do admit a ($1+\epsilon$)-spanner, its lightness depend exponentially on the dimension (resp. $\log g$). Our construction alleviates this exponential dependency, at the cost of incurring larger stretch.
\end{abstract}

\thispagestyle{empty}
\newpage
\setcounter{page}{1}
\section{Introduction}
\subsection{Spanners}

Given a metric space $(X,d_X)$, a weighted graph $H=(X,E)$ is a $t$-\emph{spanner} of $X$, if for every pair of points $x,y\in X$, $d_X(x,y)\le d_H(x,y) \le t\cdot d_X(x,y)$ (where $d_H$ is the shortest path metric in $H$).
The factor $t$ is called the \emph{stretch} of the spanner. Two important parameters of interest are: the {\em sparsity} of the spanner, i.e. the number of edges, and the \emph{lightness} of the spanner, which is the ratio between the total weight of the spanner and the weight of the minimum spanning tree (MST).



The tradeoff between stretch and sparsity/lightness of spanners is the focus of an intensive research effort, and low stretch spanners were used in a plethora of applications, to name a few: Efficient broadcast protocols \cite{ABP90,ABP91},  network synchronization \cite{Awerbuch85,PU89,ABP90,ABP91,Peleg00}, data gathering and dissemination tasks \cite{BKRCV02,VWFME03,KV01}, 
routing \cite{WCT02,PU89,PU89b,TZ01}, 
distance oracles and labeling schemes \cite{Peleg99,TZ05,RTZ05}, 
and almost shortest paths \cite{Coh98,RZ04,Elkin05,EZ06,FKMSZ05}.

Spanners for general metric spaces are well understood. The seminal paper of \cite{ADDJS93} showed that for any parameter $k\ge 1$, any metric admits a $(2k-1)$-spanner with $O(n^{1+1/k})$ edges, which is conjectured to be best possible. For light spanners, improving \cite{CDNS95,ENS14}, it was shown in \cite{CW16} that for every constant $\eps>0$ there is a $(2k-1)(1+\eps)$-spanner with lightness $O(n^{1/k})$ and at most $O(n^{1+1/k})$ edges.

There is an extensive study of spanners for restricted classes of metric spaces, most notably subsets of low dimensional Euclidean space, and more generally doubling metrics.\footnote{A metric space $(X,d)$ has doubling constant $\lambda$ if for every $x\in X$ and radius
	$r>0$, the ball $B(x,2r)$ can be covered by $\lambda$ balls of radius $r$. The doubling dimension is defined as $\ddim=\log_2\lambda$. A $d$-dimensional $\ell_p$ space has $\ddim=\Theta(d)$, and every $n$ point metric has $\ddim=O(\log n)$.}
For such low dimensional metrics, much better spanners can be obtained. Specifically, for $n$ points in $d$-dimensional Euclidean space, \cite{Sal91,Vai91,DHN93} showed that for any $\eps\in(0,\frac12)$ there is a $(1+\eps)$-spanner with $n\cdot\eps^{-O(d)}$ edges and lightness $\eps^{-O(d)}$ (further details on Euclidean spanners could be found in \cite{NS07}).
%
%
This result was recently generalized to doubling metrics by \cite{BLW17}, with  $\eps^{-O(\ddim)}$ lightness and  $n\cdot\eps^{-O(\ddim)}$ edges (improving \cite{Smid09,Got15,FS16}).
Such low stretch spanners were also devised for metrics arising from certain graph families. For instance, \cite{ADDJS93} showed that any planar graph admits a $(1+\eps)$-spanner with lightness $O(1/\eps)$. This was extended to graphs with small genus\footnote{The \emph{genus} of a graph is minimal integer $g$, such that the graph could be drawn on a surface with $g$ ``handles''.}
by \cite{Grigni00}, who showed that every graph with genus $g>0$ admits a spanner with stretch $(1+\eps)$ and lightness $O(g/\eps)$. A long sequence of works for other graph families, concluded recently with a result of \cite{BLW17focs}, who showed $(1+\eps)$-spanners for graphs excluding $K_r$ as a minor, with lightness $\approx O(r/\eps^3)$.

In all these results there is an exponential dependence on a certain parameter of the input metric space (the dimension, the logarithm of the genus/minor-size), which is unfortunately unavoidable for small stretch (for all $n$-point metric spaces the dimension/parameter is at most $O(\log n)$, while spanner with stretch better than 3 requires in general $\Omega(n^2)$ edges \cite{TZ05}). So when the relevant parameter is small, light spanners could be constructed with extremely small stretch. However, in metrics arising from actual data, the parameter of interest may be moderately large, and it is not known how to construct light spanners avoiding the exponential dependence on it.
In this paper, we devise a tradeoff between stretch and sparsity/lightness that can diminish this exponential dependence.
To the best of our knowledge, the only such tradeoff is the recent work of \cite{HPIS13}, who showed that $n$-point subsets of Euclidean space (in any dimension) admit a $O(t)$-spanner with $\tilde{O}(n^{1+1/t^2})$ edges (without any bound on the lightness).

\subsection{Stochastic Decompositions}

In a (stochastic) decomposition of a metric space, the goal is to find a partition of the points into clusters of low diameter, such that the probability of nearby points to fall into different clusters is small. More formally, for a metric space $(X,d_X)$ and parameters $t\ge 1$ and $\delta=\delta(|X|,t)\in[0,1]$, we say that the metric is $(t,\delta)$-decomposable, if for every $\Delta>0$ there is a probability distribution over partitions of $X$ into clusters of diameter at most $t\cdot \Delta$, such that every two points of distance at most $\Delta$ have probability at least $\delta$ to be in the same cluster.

Such decompositions were introduced in the setting of distributed computing \cite{A85,LS93}, and have played a major role in the theory of metric embedding \cite{Bar96,R99,FRT03,KLMN04,LN05,ABN11},  distance oracles and routing \cite{MN07,ACEFN18}, multi-commodity flow/sparsest cut gaps \cite{LR99,KPR93} and also were used in approximation algorithms and spectral methods \cite{CKR04,KLPT09,BLR10}. We are not aware of any direct connection of these decompositions to spanners (except spanners for general metrics implicit in \cite{MN07,ACEFN18}).

Note that our definition is slightly different than the standard one. The probability $\delta$ that a pair $x,y\in X$ is in the same cluster may depend on $|X|$ and $t$, but unlike previous definitions, it does not depend on the precise value of $d_X(x,y)$ (rather, only on the fact that it is bounded by $\Delta$). This simplification suits our needs, and it enables us to capture more succinctly the situation for high dimensional normed spaces, where the dependence of $\delta$ on $d_X(x,y)$ is non-linear.
These stochastic decompositions are somewhat similar to Locality Sensitive Hashing (LSH), that were used by \cite{HPIS13} to construct spanners. The main difference is that in LSH, far away points may be mapped to the same cluster with some small probability, and more focus was given to efficient computation of the hash function. It is implicit in \cite{HPIS13} that existence of good LSH imply sparse spanners.

A classic tool for constructing spanners in normed and doubling spaces is WSPD (Well Separated Pair Decomposition, see \cite{CK92,Tal04,HM06}).
Given a set of points $P$, a WSPD is a set of pairs $\{(A_i,B_i)\}_i$ of subsets of $P$, where the diameters of $A_i$ and $B_i$ are at most an $\eps$-fraction of $d(A_i,B_i)$, and such that for every pair $x,y\in P$ there is some $i$ with $(x,y)\in A_i\times B_i$.
A WSPD is designed to create a $(1+O(\eps))$-spanner, by
adding an arbitrary edge between a point in $A_i$ and a point in $B_i$ for every $i$ (as opposed to our construction, based on stochastic decompositions, in which we added only inner-cluster edges). An exponential dependence on the dimension is unavoidable with such a low stretch, thus it is not clear whether one can use a WSPD to obtain very sparse or light spanners in high dimensions.

\subsection{Our Results}

Our main result is exhibiting a connection between stochastic decompositions of metric spaces, and light spanners. Specifically, we show that if an $n$-point metric is $(t,\delta)$-decomposable, then for any constant $\eps>0$, it admits a $(2+\eps)\cdot t$-spanner with $\tilde{O}(n/\delta)$ edges and lightness $\tilde{O}(1/\delta)$. (Abusing notation, $\tilde{O}$ hides $\polylog(n)$ factors.)

It can be shown that Euclidean metrics are $(t,n^{-O(1/t^2)})$-decomposable, thus our results extends \cite{HPIS13} by providing a smaller stretch $(2+\eps)\cdot t$-spanner, which is both sparse -- with $\tilde{O}(n^{1+O(1/t^2)})$ edges -- and has lightness $\tilde{O}(n^{O(1/t^2)})$. For $d$-dimensional Euclidean space, where $d=o(\log n)$ we can obtain lightness $\tilde{O}(2^{O(d/t^2)})$ and $\tilde{O}(n\cdot 2^{O(d/t^2)})$ edges.
We also show that $n$-point subsets of $\ell_p$ spaces for any fixed $1<p< 2$ are $(t,n^{-O(\log^2t/t^p)})$-decomposable, which yields light spanners for such metrics as well.

In addition, metrics with doubling constant $\lambda$ are $(t,\lambda^{-O(1/t)})$-decomposable \cite{GKL03,ABN11}, and graphs with genus $g$ are $(t,g^{-O(1/t)})$-decomposable \cite{LS10,AGGNT14}, which enables us to alleviate the exponential dependence on $\ddim$ and $\log g$ in the sparsity/lightness by increasing the stretch. See \tableref{tab:results} for more details. (We remark that for graphs excluding $K_r$ as a minor, the current best decomposition achieves probability only $2^{-O(r/t)}$ \cite{AGGNT14}; if this will be improved to the conjectured $r^{-O(1/t)}$, then our results would provide interesting spanners for this family as well.)





\begin{table}[!ht]
	\centering	
	\begin{tabular}{|c|l|l|l|l|}
		\hline
		\multicolumn{1}{|l|}{}                      & \multicolumn{1}{c|}{\textbf{Stretch}}                 & \multicolumn{1}{c|}{\textbf{Lightness}} & \multicolumn{1}{c|}{\textbf{Sparsity}}  &  \\ \hline
		\multirow{2}{*}{Euclidean space}            & $O(t)$                                                & $\tilde{O}(n^{\sfrac{1}{t^2}})$          & $\tilde{O}(n^{1+\sfrac{1}{t^2}})$      & \multirow{2}{*}{\corollaryref{cor:L2}}  \\ \cline{2-4}
		& $O(\sqrt{\log n})$                 & $\tilde{O}(1)$                           & $\tilde{O}(n)$                    & \\ \hline
		\multirow{2}{*}{$\ell_p$ space, $1<p<2$}             & $O(t)$                                                & $\tilde{O}(n^{\sfrac{\log^2 t}{t^p}})$   & $\tilde{O}(n^{1+\sfrac{\log^2 t}{t^p}})$ & \multirow{2}{*}{\corollaryref{cor:Lp}} \\ \cline{2-4}
		& $O((\log n\cdot\log\log n)^{\sfrac{1}{p}})$ & $\tilde{O}(1)$                           & $\tilde{O}(n)$                   & \\ \hline
		\multirow{2}{*}{Doubling constant $\lambda$} & $O(t)$                                                & $\tilde{O}(\lambda^{\sfrac{1}{t}})$   & $\tilde{O}(n\cdot \lambda^{\sfrac{1}{t}})$ &\multirow{2}{*}{\corollaryref{cor:ddim}} \\ \cline{2-4}
		& $O(\log\lambda)$                       & $\tilde{O}(1)$                           & $\tilde{O}(n)$                    & \\ \hline
		\multirow{2}{*}{Graph with genus $g$}       & $O(t)$                                                & $\tilde{O}(g^{\sfrac1t})$           & $O(n+g)$         & \multirow{2}{*}{\corollaryref{cor:genus}} \\ \cline{2-4}
		& $O(\log g)$                      & $\tilde{O}(1)$                           & $O(n+g)$                  & \\ \hline
	\end{tabular}
	\caption{In this table we summarize some corollaries of our main result. The metric spaces have cardinality $n$, and $\tilde{O}$ hides (mild) $\polylog(n)$ factors.
		The stretch $t$ is a parameter ranging between $1$ and $\log n$.
	}\label{tab:results}
\end{table}

%
Note that up to $\polylog(n)$ factors, our stretch-lightness tradeoff generalizes the \cite{CW16} spanner for general metrics, which has stretch $(2t-1)(1+\eps)$ and lightness $O(n^{1/t})$. Define for a $(t,\delta)$-decomposable metric the parameter $\nu=1/\delta^t$. Then
we devise for such a metric a $(2t-1)(1+\eps)$-spanner with lightness $O(\nu^{1/t})$.

For example, consider an $n$-point metric with doubling constant $\lambda=2^{\sqrt{\log n}}$. No spanner with stretch $o(\log n/\log\log n)$ and lightness $\tilde{O}(1)$ for such a metric was known. Our result imply such a spanner, with stretch $O(\sqrt{\log n})$.

We also remark that the existence of light spanners does not imply decomposability. For example, consider the shortest path metrics induced by bounded-degree expander graphs. Even though these metrics have the (asymptotically) worst possible decomposability parameters (they are only $(t,n^{-\Omega(1/t)})$-decomposable \cite{llr}), they nevertheless admit $1$-spanners with constant lightness (the spanner being the expander graph itself).

\section{Preliminaries}


Given a metric space $(X,d_X)$, let $T$ denote its minimum spanning tree (MST) of weight $L$. For a set $A\subseteq X$, the diameter of $A$ is $\diam(A)=\max_{x,y\in A}d_X(x,y)$.
Assume, as we may, that the minimal distance in $X$ is $1$.

By $O_\eps$ we denote asymptotic notation which hides polynomial factors of $\frac1\eps$, that
is $O_\eps(f)=O(f)\cdot\poly(\frac1\eps)$. Unless explicitly specified otherwise, all logarithms are in base $2$.


\paragraph{Nets.}
For $r>0$, a set $N\subseteq X$ is an $r$-\emph{net}, if (1) for every $x\in X$ there is a point $y\in N$ with $d_X(x,y)\le r$, and (2) every pair of net points $y,z\in N$ satisfy $d_X(y,z)>r$.
It is well known that nets can be constructed in a greedy manner. For $0<r_1\le r_2\le\dots\le r_s$, a {\em hierarchical net} is a collection of nested sets $X\supseteq N_1\supseteq N_2\supseteq\dots\supseteq N_s$, where each $N_i$ is an $r_i$-net. Since $N_{i+1}$ satisfies  the second condition of a net with respect to radius $r_i$, one can obtain $N_i$ from $N_{i+1}$ by greedily adding points until the first condition is satisfied as well.
In the following claim we argue that nets are sparse sets with respect to the MST weight.

\begin{claim}\label{claim:NetSize}
	Consider a metric space $(X,d_X)$ with MST of weight $L$, let $N$ be an $r$-net, then $|N|\le\frac{2L}{r}$.
\end{claim}
\begin{proof}
	Let $T$ be the MST of $X$, note that for every $x,y\in N$, $d_T(x,y)\ge d_X(x,y)> r$.
	For a point $x\in N$, $B_T(x,b)=\{y\in X\mid d_T(x,y)\le b\}$ is the ball of radius $b$ around $x$ in the MST metric. We say that an edge $\{y,z\}$ of $T$ is {\em cut} by the ball $B_T(x,b)$ if $d_T(x,y)<b<d_T(x,z)$.
	Consider the set $\mathcal{B}$ of balls of radius $r/2$ around
	the points of $N$.
	We can subdivide\footnote{To subdivide an edge $e=\{x,y\}$ of weight $w$ the following steps are taken: (1) Delete the edge $e$. (2) Add a new vertex $v_e$. (3) Add two new edges $\{x,v_e\},\{v_e,y\}$ with weights $\alpha\cdot w$ and $(1-\alpha)\cdot w$ for some $\alpha\in(0,1)$.}
	the edges of $T$ until no edge is cut by any of the balls of $\mathcal{B}$. Note that the subdivisions do not change the total weight of $T$ nor the distances between the original points of $X$.	
	
	If both the endpoints of an edge $e$ belong to the ball $B$, we say that the edge $e$ is internal to $B$. By the second property of nets, and since $B_T(x,b)\subseteq B_X(x,b)$, the set of internal edges corresponding to the balls $\mathcal{B}$ are disjoint. On the other hand, as the tree is connected, the weight of the internal edges in each ball must be at least $r/2$. As the total weight is bounded by $L$, the claim follows.
\end{proof}

\paragraph{Stochastic Decompositions.}
Consider a \emph{partition} $\mathcal{P}$ of $X$ into disjoint clusters.
For $x\in X$, we denote by $\mathcal{P}(x)$ the cluster $P\in \mathcal{P}$ that contains $x$.
A partition $\mathcal{P}$ is $\Delta$-\emph{bounded} if for every $P\in\mathcal{P}$, $\diam(P)\le\Delta$.
If a pair of points $x,y$ belong to the same cluster, i.e. $\mathcal{P}(x)=\mathcal{P}(y)$, we say that they are {\em clustered} together by $\mathcal{P}$.
\begin{definition}\label{def:decompostion}
	For metric space $\left(X,d_{X}\right)$ and parameters $t\ge 1$, $\Delta>0$ and $\delta\in[0,1]$, a distribution
	$\mathcal{D}$ over partitions of $X$ is called a $(t,\Delta,\delta)$-decomposition,
	if it fulfills the following properties.
	\begin{itemize}
		\item Every $\mathcal{P}\in\text{supp}(\mathcal{D})$ is $t\cdot\Delta$-bounded.
		\item For every $x,y\in X$ such that $d_{X}(x,y)\le \Delta$, $\Pr_{\mathcal{D}}\left[\mathcal{P}(x)=\mathcal{P}(y)\right]\ge \delta$.
	\end{itemize}
	A metric is $(t,\delta)$-decomposable, where $\delta=\delta(|X|,t)$, if it admits a $(t,\Delta,\delta)$-decomposition for any $\Delta>0$. A family of metrics is $(t,\delta)$-decomposable if each member $(X,d_X)$ in the family is $(t,\delta)$-decomposable.
\end{definition}


We observe that if a metric $(X,d_X)$ is $(t,\delta(|X|,t))$-decomposable, then also every sub-metric $Y\subseteq X$ is $(t,\delta(|X|,t))$-decomposable. In some cases $Y$ is also $(t,\delta(|Y|,t))$-decomposable (we will exploit these improved decompositions for subsets of $\ell_p$).
The following claim argues that sampling $O(\frac{\log n}{\delta})$ partitions suffices to guarantee that every pair is clustered at least once.
\begin{claim}\label{claim:CoveringUsingDecopositions}
	Let $\left(X,d_{X}\right)$ be a metric
	space which admits a $\left(t,\Delta,\delta\right)$-decomposition, and let $N\subseteq X$ be of size $|N|=n$.
	Then there is a set $\left\{ \mathcal{P}_{1},\dots,\mathcal{P}_{\varphi}\right\}$ of $t\cdot\Delta$-bounded  partitions of $N$, where $\varphi=\frac{2\ln n}{\delta}$,
	such that every pair $x,y\in N$ at distance at most $\Delta$ is
	clustered together by at least one of the $\mathcal{P}_{i}$.
\end{claim}
\begin{proof}
	Let $\left\{ \mathcal{P}_{1},\dots,\mathcal{P}_{\varphi}\right\}$ be i.i.d  partitions drawn from the $\left(t,\Delta,\delta\right)$-decomposition of $X$. Consider a pair $x,y\in N$ at distance at most $\Delta$. The probability that $x,y$ are not clustered in any of the partitions is bounded by
	\[
	\Pr\left[\forall i,~~\mathcal{P}_{i}(x)\ne\mathcal{P}_{i}(y)\right]~\le~\left(1-\delta\right)^{(2\ln n)/\delta}~\le~\frac{1}{n^{2}}~.
	\]
	The claim now follows by the union bound.
\end{proof}

\section{Light Spanner Construction}
In this section we present a generalized version of the algorithm of \cite{HPIS13}, depicted in \algref{alg:LightSpannerDecomp}. The differences in execution and analysis are: (1) Our construction applies to general decomposable metric spaces -- we use decompositions rather than LSH schemes. (2) We analyze the lightness of the resulting spanners. (3) We achieve stretch $t\cdot(2+\eps)$ rather than $O(t)$. 

The basic idea is as follows.
For every weight scale $\Delta_i=(1+\eps)^i$, construct a sequence of $t\cdot\Delta_i$-bounded partitions $\mathcal{P}_{1},\dots,\mathcal{P}_{\varphi}$ such that every pair $x,y$ at distance $\le\Delta_i$ will be clustered together at least once.
Then, for each $j\in[\varphi]$ and every cluster $P\in\mathcal{P}_j$, we pick an arbitrary root vertex $v_P\in P$, and add to our spanner edges from $v_P$ to all the points in $P$.
This ensures stretch $2t\cdot(1+\eps)$ for all pairs with $d_X(x,y)\in[(1-\eps)\Delta_i,\Delta_i]$. Thus, repeating this procedure on all scales $i=1,2,\dots$ provides a spanner with stretch $2t\cdot(1+\eps)$.

However, the weight of the spanner described above is unbounded. In order to address this problem at scale $\Delta_i$, instead of taking the partitions over all points, we partition only the points of an $\eps\Delta_i$-net. The stretch is still small: $x,y$ at distance $\Delta_i$ will have nearby net points $\tilde{x},\tilde{y}$. Then, a combination of newly added edges with older ones will produce a short path between $x$ to $y$. The bound on the lightness will follow from the observation that the number of net points is bounded with respect to the MST weight.

\begin{theorem}\label{thm:LightSpannerRegDecomp}
	Let $\left(X,d_{X}\right)$ be a $(t,\delta)$-decomposable $n$-point metric space.
	Then for every $\eps\in(0,1/8)$, there is a $t\cdot(2+\eps)$-spanner for $X$ with lightness $\Oeps\left(\frac{t}{\delta}\cdot\log^{2}n\right)$
	and $\Oeps\left(\frac n\delta\cdot\log n\cdot\log t\right)$ edges.
\end{theorem}
\begin{algorithm}[h]
	\caption{$H=\texttt{Spanner-From-Decompositions}((X,d_X),t,\eps)$}\label{alg:LightSpannerDecomp}
	\begin{algorithmic}[1]
		\STATE Let $N_0\supseteq N_1\supseteq\cdots\supseteq N_{\log_{1+\eps}L}$ be a hierarchical net, where $N_i$ is $\eps\cdot \Delta_i=\eps\cdot(1+\eps)^i$-net of $(X,d_X)$.
		\FOR {$i\in\left\{ 0,1,\dots,\log_{1+\eps} L\right\} $} \label{line:ForScale}
		\STATE For parameters $\Delta=(1+2\eps)\Delta_i$ and $t$, let
		$\mathcal{P}_{1},\dots,\mathcal{P}_{\varphi_{i}}$ be the set of $t\cdot\Delta$-bounded partitions guaranteed by \claimref{claim:CoveringUsingDecopositions} on the set $N_i$.
		\label{line:createPartitions}
		\FOR {$j\in\left\{ 1,\dots,\varphi_{i}\right\} $ and $P\in\mathcal{P}_{j}$}
		\STATE Let $v_P\in P$ be an arbitrarily point. \label{line:chooseVp}
		\STATE Add to $H$ an edge from every point  $x\in P\setminus\{v_P\}$ to $v_P$. \label{line:AddStar}
		\ENDFOR		
		\ENDFOR
		\RETURN $H$.
	\end{algorithmic}	
\end{algorithm}

\begin{proof}
	We will prove stretch $t\cdot (2+O(\eps))$ instead of $t\cdot(2+\eps)$. This is good enough, as post factum we can scale $\eps$ 	 accordingly.
	
	\paragraph{Stretch Bound.}	
	Let $c>1$ be a constant (to be determined later).
	Consider a pair $x,y\in X$ such that $(1+\eps)^{i-1}<d_{X}(x,y)\le(1+\eps)^{i}$.
	We will assume by induction that every pair $x',y'$ at distance at most $(1+\eps)^{i-1}$ already enjoys stretch at most $\alpha=t\cdot(2+c\cdot\eps)$ in $H$.
	Set $\Delta_i=(1+\eps)^{i}$, and let $\tilde{x},\tilde{y}\in N_{i}$ be net points such that $d_{X}(x,\tilde{x}),d_{X}(y,\tilde{y})\le\eps\cdot \Delta_i$.
	By the triangle inequality $d_{X}(\tilde{x},\tilde{y})\le (1+2\eps)\cdot\Delta_i=\Delta$.
	Therefore there is a $t\cdot\Delta$-bounded partition $\mathcal{P}$ constructed at round $i$ such that $\mathcal{P}(\tilde{x})=\mathcal{P}(\tilde{y})$.
	In particular, there is a center vertex $v=v_{\mathcal{P}(\tilde{x})}$ such that both $\left\{ \tilde{x},v\right\} ,\left\{ \tilde{y},v\right\} $
	were added to the spanner $H$. Using the induction hypothesis on the pairs $\{x,\tilde{x}\}$ and $\{y,\tilde{y}\}$, we conclude
	\begin{align*}
	d_{H}\left(x,y\right) & \le d_{H}\left(x,\tilde{x}\right)+d_{H}\left(\tilde{x},v\right)+d_{H}\left(v,\tilde{y}\right)+d_{H}\left(\tilde{y},y\right)\\
	& \le\alpha\cdot\epsilon\Delta_i+(1+2\epsilon)t\Delta_i+(1+2\epsilon)t\Delta_i+\alpha\cdot\epsilon\Delta_i\\
	& \overset{(*)}{<}\frac{\alpha}{1+\epsilon}\cdot\Delta_i\le\alpha\cdot d_{X}\left(x,y\right)~,
	\end{align*}
	where the inequality $(*)$ follows as $2(1+2\eps)t<\alpha(\frac{1}{1+\eps}-2\eps)$ for large enough constant $c$, using that $\eps<1/8$.

	\paragraph{Sparsity bound.}
	For a point $x\in X$, let $s_x$ be the maximal index such that $x\in N_{s_x}$. Note that the number of edges in our spanner is not affected by the choice of ``cluster centers" in line \ref{line:chooseVp} in \algref{alg:LightSpannerDecomp}. Therefore, the edge count will be still valid if we assume that $v_P\in P$ is the vertex $y$ with maximal value $s_y$ among all vertices in $P$.
	
	Consider an edge $\{x,y\}$ added during the $i$'s phase of the algorithm. Necessarily $x,y\in N_i$, and $x,y$ belong to the same cluster $P$ of a partition $\mathcal{P}_j$. W.l.o.g, $y=v_P$, in particular $s_x\le s_y$. The edge $\{x,y\}$ will be charged upon $x$.
	Since the partitions at level $i$ are $t\cdot\Delta$ bounded, we have that $d_X(x,y)\le t\cdot\Delta=t\cdot(1+2\eps)\cdot(1+\eps)^i$.
	Hence, for $i'$ such that $\eps\cdot(1+\eps)^{i'}>t\cdot(1+2\eps)\cdot(1+\eps)^i$, i.e. $i'>i+\Oeps(\log t)$,
	the points $x,y$ cannot both belong to $N_{i'}$. As $s_x\le s_y$, it must be that $x\notin N_{i'}$.
	We conclude that $x$ can be charged in at most $\Oeps\left(\log t\right)$ different levels. As in level $i$ each vertex is charged for at most $\varphi_i\le O(\frac{\log n}{\delta})$ edges, the total charge for each vertex is bounded by $\Oeps(\frac{\log n\cdot \log t}{\delta})$.

	\paragraph{Lightness bound.}
	Consider the scale $\Delta_i=(1+\eps)^i$. As $N_i$ is an $\eps\cdot\Delta_i$-net, \claimref{claim:NetSize} implies that $N_i$ has size $n_i\le\frac{2L}{\eps\cdot\Delta_i}$, and in any case at most $n$.
	In that scale, we constructed $\varphi_i=\frac2\delta\log n_i\le\frac2\delta\log n$ partitions, adding at most $n_{i}$ edges per partition. The weight of each edge added in this scale is bounded by $O(t\cdot \Delta_i)$.
	
	Let $H_1$ consist of all the edges added in scales $i\in\{\log_{1+\eps}\frac Ln,\dots, \log_{1+\eps}L\}$, while $H_2$ consist of edges added in the lower scales. Note that $H=H_1\cup H_2$.
	\begin{align*}
	w\left(H_{1}\right) & \le\sum_{i\in\left\{ \log_{1+\epsilon}\frac{L}{n},\dots,\log_{1+\epsilon}L\right\} }O\left(t\cdot\Delta_{i}\right)\cdot n_{i}\cdot\varphi_{i}\\
	& =O\left(\frac{t}{\delta}\cdot\log n\cdot\sum_{i\in\left\{ \log_{1+\epsilon}\frac{L}{n},\dots,\log_{1+\epsilon}L\right\} }\Delta_{i}\cdot\frac{L}{\epsilon\cdot\Delta_{i}}\right)=\Oeps\left(\frac{t}{\delta}\cdot\log^{2}n\right)\cdot L~.\\
	w\left(H_{2}\right) & \le\sum_{\Delta_{i}\in\frac{L}{n}\cdot\left\{ (1+\epsilon)^{-1},(1+\epsilon)^{-2},\dots,\right\} }O\left(t\cdot\Delta_{i}\right)\cdot n_{i}\cdot\varphi_{i}\\
	& =O\left(\frac{t}{\delta}\cdot\log n\cdot\sum_{i\ge1}\frac{1}{(1+\epsilon)^{i}}\right)\cdot L=\Oeps\left(\frac{t}{\delta}\cdot\log n\right)\cdot L~.
	\end{align*}
	The bound on the lightness follows.
\end{proof}

\section{Corollaries and Extensions}

In this section we describe some corollaries of \theoremref{thm:LightSpannerRegDecomp} for certain metric spaces, and show some extensions, such as improved lightness bound for normed spaces, and discuss graph spanners.

\subsection{High Dimensional Normed Spaces}

Here we consider the case that the given metric space $(X,d)$ satisfies that every sub-metric $Y\subseteq X$ of size $|Y|=n$ is $(t,\delta)$-decomposable for $\delta=n^{-\beta}$, where $\beta=\beta(t)\in (0,1)$ is a function of $t$.
In such a case we are able to shave a $\log n$ factor in the lightness.

\begin{theorem}\label{thm:LightSpannerPropDecomp}
	Let $\left(X,d_{X}\right)$ be an $n$-point metric space such that every $Y\subseteq X$ is $(t,|Y|^{-\beta})$-decomposable.
	Then for every $\eps\in(0,1/8)$, there is a $t\cdot(2+\eps)$-spanner for $X$ with lightness $\Oeps\left(\frac{t}{\beta}\cdot n^{\beta}\cdot\log n\right)$ and sparsity $\Oeps\left(n^{1+\beta}\cdot\log n\cdot\log t\right)$.
\end{theorem}
\begin{proof}
	Using the same \algref{alg:LightSpannerDecomp}, the analysis of the stretch and sparsity from \theoremref{thm:LightSpannerRegDecomp} is still valid, since the number partitions taken in each scale is smaller than in
	\theoremref{thm:LightSpannerRegDecomp}. Recall that in scale $i$ we set $\Delta_i=(1+\eps)^i$, and the size of the $\eps\cdot\Delta_i$-net $N_i$ is $n_i\le\max\{\frac{2L}{\eps\Delta_i},n\}$. The difference from the previous proof is that $N_i$ is $(t,n_i^{-\beta})$-decomposable, so the number of partitions taken is $\varphi_i=O(n_i^\beta\log n_i)$. In each partition we might add at most one edge per net point, and the weight of this edge is $O(t\cdot\Delta_i)$.
	We divide the edges of $H$ to $H_1$ and $H_2$, and bound the weight of $H_2$ as above (using that $n_i\le n$). For $H_1$ we get,
	
	\begin{align*}
	w\left(H_{1}\right) & \le\sum_{i\in\left\{ \log_{1+\epsilon}\frac{L}{n},\dots,\log_{1+\epsilon}L\right\} }O\left(t\cdot\Delta_{i}\right)\cdot n_{i}\cdot\varphi_{i}\\
	& =O\left(t\cdot\sum_{i\in\left\{ \log_{1+\epsilon}\frac{L}{n},\dots,\log_{1+\epsilon}L\right\} }\Delta_{i}\cdot\frac{L}{\epsilon\cdot\Delta_{i}}\cdot\left(\frac{L}{\epsilon\cdot\Delta_{i}}\right)^{\beta}\log\frac{L}{\epsilon\cdot\Delta_{i}}\right)\\
	& =\Oeps\left(t\cdot\sum_{i\in\left\{ \log_{1+\epsilon}\frac{L}{n},\dots,\log_{1+\epsilon}L\right\} }\left(\frac{L}{\Delta_{i}}\right)^{\beta}\cdot\log\frac{L}{\Delta_{i}}\right)\cdot L\\
	& =\Oeps\left(t\cdot\sum_{i\in\left\{ 0,\dots,\log_{1+\epsilon}n\right\} }\left(i+1\right)\cdot\left(\left(1+\epsilon\right)^{\beta}\right)^{i}\right)\cdot L~.
	\end{align*}
	Set the function $f(x)=\sum_{i=0}^{k}\left(i+1\right)\cdot x^{i}$, on the domain $(1,\infty)$, with parameter $k=\log_{1+\epsilon}n$. Then,
	\begin{align*}
	f(x)=\left(\int fdx\right)' & =\left(\sum_{i=0}^{k}x^{i+1}\right)'=\left(\frac{x^{k+2}-x}{x-1}\right)'\\
	& =\frac{\left(\left(k+2\right)x^{k+1}-1\right)\left(x-1\right)-\left(x^{k+2}-x\right)}{\left(x-1\right)^{2}}\le\frac{\left(k+2\right)x^{k+1}}{x-1}~.
	\end{align*}
	Hence,
	\begin{align*}
	w\left(H_{1}\right) & =\Oeps\left(t\cdot f\left((1+\epsilon)^{\beta}\right)\right)\cdot L\\
	& =\Oeps\left(t\cdot\frac{\log_{1+\epsilon}n\cdot\left(\left(1+\epsilon\right)^{\beta}\right)^{\log_{1+\epsilon}n}}{\left(1+\epsilon\right)^{\beta}-1}\right)\cdot L=\Oeps\left(\frac{t}{\beta}\cdot n^{\beta}\cdot\log n\right)\cdot L~.
	\end{align*}
	
	
	We conclude that the lightness of $H$ is bounded by $\Oeps\left(\frac{t}{\beta}\cdot n^{\beta}\cdot\log n\right)$.
\end{proof}

In \sectionref{app:LSHtoDecomp} we will show that any $n$-point Euclidean metric is $(t,n^{-O(\sfrac{1}{t^2})})$-decomposable,
and that for fixed $p\in(1,2)$, any $n$-point subset of $\ell_p$ is $(t,n^{-O(\sfrac{\log^2 t}{t^p})})$-decomposable. The following corollaries are implied by \theoremref{thm:LightSpannerPropDecomp} (rescaling $t$ by a constant factor allows us to remove the $O(\cdot)$ term in the exponent of $n$, while obtaining stretch $O(t)$).
\begin{corollary}\label{cor:L2}
	For a set $X$ of $n$ points in Euclidean space, $t>1$, there is an $O(t)$-spanner with lightness $O\left(t^{3}\cdot n^{\sfrac{1}{t^{2}}} \cdot\log n\right)$
	and $O\left(n^{1+\sfrac{1}{t^{2}}}\cdot\log n\cdot\log t\right)$ edges.
\end{corollary}

%
\begin{corollary}\label{cor:Lp}
	For a constant $p\in(1,2)$ and a set $X$ of $n$ points in $\ell_{p}$
	space, there is an $O(t)$-spanner with lightness
	$O\left(\frac{t^{1+p}}{\log^{2}t}\cdot n^{\sfrac{\log^{2}t}{t^{p}}}\cdot\log n\right)$
	and $O\left(n^{1+\sfrac{\log^{2}t}{t^{p}}}\cdot\log n\cdot\log t\right)$ edges.
\end{corollary}

\begin{remark}
	\corollaryref{cor:L2} applies for a set of points $X\subseteq\mathbb{R}^d$, where the dimension $d$ is arbitrarily large.
	If $d=o(\log n)$ we can obtain improved spanners. Specifically, $n$-point subsets of $d$-dimensional Euclidean space are $(O(t),2^{-\sfrac{d}{t^{2}}})$-decomposable (see \sectionref{app:dDimension}). Applying \theoremref{thm:LightSpannerRegDecomp} we obtain an $O(t)$-spanner with lightness $O_{\eps}\left(t\cdot2^{\sfrac{d}{t^{2}}}\cdot\log^{2}n\right)$ and  $O_{\eps}\left(n\cdot2^{\sfrac{d}{t^{2}}}\cdot\log n\cdot\log t\right)$ edges. 
\end{remark}

\subsection{Doubling Metrics}

It was shown in \cite{ABN11} that metrics with doubling constant $\lambda$ are $(t,\lambda^{-O(\sfrac{1}{t})})$-decomposable (the case $t=\Theta(\log\lambda)$ was given by \cite{GKL03}). 
Therefore, \theoremref{thm:LightSpannerRegDecomp} implies:
\begin{corollary}\label{cor:ddim}
	For every metric space $(X,d_X)$ with doubling constant $\lambda$, and $t\ge1$, there exist an $O(t)$-spanner with lightness $O\left(t\cdot\log^{2}n\cdot \lambda^{\sfrac{1}{t}}\right)$ and $O\left(n\cdot \lambda^{\sfrac{1}{t}}\cdot\log n\cdot\log t\right)$ edges.
\end{corollary}

\subsection{Graph Spanners}
In the case where the input is a graph $G$, it is natural to require that the spanner will be a {\em graph-spanner}, i.e., a subgraph of $G$. Given a (metric) spanner $H$, one can define a graph-spanner $H'$ by replacing every edge $\{x,y\}\in H$ with the shortest path from $x$ to $y$ in $G$. It is straightforward to verify that the stretch and lightness of $H'$ are no larger than those of $H$ (however, the number of edges may increase).

Consider a graph $G$ with genus $g$. In \cite{AGGNT14} it was shown that (the shortest path metric of) $G$ is $\left(t,g^{-O(\sfrac1t)}\right)$-decomposable. Furthermore, graphs with genus $g$ have $O(n+g)$ edges \cite{GT87}, so any graph-spanner will have at most so many edges. By \theoremref{thm:LightSpannerRegDecomp} we have:
\begin{corollary}\label{cor:genus}
	Let $G$ be a weighted graph on $n$ vertices with genus $g$. Given a parameter $t\ge1$, there exist an $O(t)$-graph-spanner of $G$ with lightness $O\left(t\cdot\log^{2}n\cdot g^{\sfrac{1}{t}}\right)$ and $O(n+g)$ edges.
\end{corollary}

For general graphs, the transformation to graph-spanners described above may arbitrarily increase the number of edges (in fact, it will be bounded by $O(\sqrt{|E_H|}\cdot n)$, \cite{CE06}). Nevertheless, if we have a \emph{strong-decomposition}, we can modify \algref{alg:LightSpannerDecomp} to produce a sparse spanner.
In a graph $G=(X,E)$, the \emph{strong-diameter} of a cluster $A\subseteq X$ is $\max_{v,u\in A}d_{G[A]}(v,u)$, where $G[A]$ is the induced graph by $A$ (as opposed to weak diameter, which is computed w.r.t the original metric distances).
A partition $\mathcal{P}$ of $X$ is $\Delta$\emph{-strongly-bounded} if the strong diameter of every $P\in\mathcal{P}$ is at most $\Delta$.
A distribution $\mathcal{D}$ over partitions of $X$ is $\left(t,\Delta,\delta\right)$-\emph{strong}-decomposition, if it is $\left(t,\Delta,\delta\right)$-decomposition and in addition every partition $\mathcal{P}\in\supp(\mathcal{D})$ is $\Delta$-strongly-bounded.
A graph $G$ is $(t,\delta)$-\emph{strongly-decomposable}, if for every $\Delta>0$, the graph admits a $\left(\Delta,t\cdot \Delta,\delta\right)$-strong-decomposition.
\begin{theorem}\label{lem:GraphSpanner}
	Let $G=\left(V,E,w\right)$ be a  $(t,\delta)$-strongly-decomposable, $n$-vertex graph with aspect ratio $\Lambda=\frac{\max_{e\in E}w(e)}{\min_{e\in E}w(e)}$.
	Then for every $\eps\in(0,1)$, there is a $t\cdot(2+\eps)$-graph-spanner for $G$ with lightness 	 $\Oeps\left(\frac{t}{\delta}\cdot\log^{2}n\right)$ and $\Oeps(\frac n\delta\cdot\log n\cdot\log \Lambda)$ edges.
\end{theorem}
\begin{proof}
	We will execute \algref{alg:LightSpannerDecomp} with several modifications:
	\begin{enumerate}
		\item The for loop (in \lineref{line:ForScale}) will go over scales $i\in\{0,\dots,\log_{1+\eps}\Lambda\}$ (instead $\{0,\dots,\log_{1+\eps}L\}$).
		\item We will use strong-decompositions instead of regular (weak) decompositions.
		\item The partitions created in \lineref{line:createPartitions} will be over the set of all vertices $V$, rather then only net points $N_i$ (as otherwise it will be impossible to get strong diameter).\\
		However, the requirement from close pairs to be clustered together (at least once), is still applied to net points only.	
		Similarly to \claimref{claim:CoveringUsingDecopositions}, $\varphi_i=(2\ln n_i)/\delta$ repetitions will suffice.
		\item In \lineref{line:AddStar}, we will no longer add edges from $v_P$ to all the net points in $P\in\mathcal{P}_j$. Instead, for every net point $x\in P\cap N_i$, we will add a shortest path in $G[P]$ from $v_P$ to $x$. Note that all the edges added in all the clusters 	
		constitute a forest. Thus we add at most $n$ edges per partition.
	\end{enumerate}
	We now prove the stretch, sparsity and lightness of the resulting spanner.
	
	\paragraph{Stretch.}
	By the triangle inequality, it is enough to show small stretch guarantee only for edges (that is, only for $x,y\in V$ s.t. $\{x,y\}\in E$.) As we assumed that the minimal distance is $1$, all the weights are within $[1,\Lambda]$. In particular, every edge $\{x,y\}\in E$ has weight $(1+\eps)^{i-1}< w\le (1+\eps)^i$ for $i\in\{0,\dots,\log_{1+\eps}\Lambda\}$.
	The rest of the analysis is similar to \theoremref{thm:LightSpannerRegDecomp}, with the only difference being that we use a path from $v_P$ to $\tilde{x}$ rather than the edge $\{\tilde{x},v_P\}$. This is fine since we only require that the length of this path is at most $(t\cdot (1+2\eps)\cdot \Delta)$, which is guaranteed by the strong diameter of clusters.
	
	\paragraph{Sparsity.} We have $\Oeps(\log\Lambda)$ scales. In each scale we had at most $\varphi_i\le\frac2\delta\log n$ partitions, where for each partition we added at most $n$ edges. The bound on the sparsity follows.
	
	\paragraph{Lightness.} Consider scale $i$. We have $n_i$ net points. For each net point we added at most one shortest path of weight at most $O(t\cdot \Delta_i)$ (as each cluster is $O(t\cdot \Delta_i)$-strongly bounded). As the number of partitions is $\varphi_i$, the total weight of all edges added at scale $i$ is bounded by $O(t\cdot\Delta_i)\cdot n_i\cdot \varphi_i$. The rest of the analysis follows by similar lines to \theoremref{thm:LightSpannerRegDecomp} (noting that $\Lambda < L$).
\end{proof}

\section{LSH Induces Decompositions}\label{app:LSHtoDecomp}
In this section, we prove that LSH (locality sensitive hashing) induces decompositions. In particular, using the LSH schemes of \cite{AI06,Ngu13}, we will get decompositions for $\ell_2$ and $\ell_p$ spaces, $1<p<2$. 
\begin{definition}
	(Locality-Sensitive-Hashing) Let $H$ be a family of hash functions
	mapping a metric $\left(X,d_{X}\right)$ to some universe $U$. We
	say that $H$ is $\left(r,cr,p_{1},p_{2}\right)$-sensitive if for
	every pair of points $x,y\in X$, the following properties are satisfied:
	\begin{enumerate}
		\item If $d_{X}(x,y)\le r$ then $\Pr_{h\in H}\left[h(x)=h(y)\right]\ge p_{1}$.
		\item If $d_{X}(x,y)>cr$ then $\Pr_{h\in H}\left[h(x)=h(y)\right]\le p_{2}$.
	\end{enumerate}
\end{definition}

Given an LSH, its parameter is $\rho=\frac{\log\sfrac{1}{p_{1}}}{\log\sfrac{1}{p_{2}}}$.
We will implicitly always assume that $p_{1}\ge n^{-\rho}$ ($n=|X|$),
as indeed will occur in all the discussed settings.
Andoni and Indyk \cite{AI06} showed that for Euclidean space ($\ell_2$), and large enough $t>1$, there is an LSH with parameter $\rho=O\left(\frac{1}{t^{2}}\right)$.
Nguyen \cite{Ngu13}, showed that for constant $p\in(1,2)$, and large enough $t>1$, there is an LSH for $\ell_p$, with parameter $\rho=O\left(\frac{\log^{2}t}{t^{p}}\right)$.
We start with the following claim.
\begin{claim}\label{claim:hash}
	Let $\left(X,d_{X}\right)$ be a metric space, such that for every $r>0$, there is an $(r,t\cdot r,p_1,p_2)$-sensitive LSH family with parameter $\rho$.
	Then there is an $\left(r,t\cdot r,n^{-O(\rho)},n^{-2}\right)$-sensitive LSH family for $X$.
\end{claim}	
\begin{proof}
	Set $k=\left\lceil \log_{\frac{1}{p_{2}}}n^{2}\right\rceil \le\frac{O(\log n)}{\log\frac{1}{p_{2}}}$, and let $H$ be the promised $(r,t\cdot r,p_1,p_2)$-sensitive LSH family.
	We define an LSH family $H'$ as follows. In order to sample $h\in H'$, pick $h_{1},\dots,h_{k}$ uniformly and independently at random from $H$. The hash function $h$ is defined as the concatenation of $h_1,\dots,h_k$. That is, $h(x)=\left(h_{1}(x),\dots,h_{k}(x)\right)$.\\
	For $x,y\in X$ such that $d_{X}(x,y)\ge t\cdot r$ it holds that
	\[
	\Pr\left[h(x)=h(y)\right]=\Pi_{i}\Pr\left[h_{i}(x)=h_{i}(y)\right]\le p_{2}^{k}\le n^{-2}~.
	\]
	On the other hand, for $x,y\in X$ such that $d_{X}(x,y)\le r$,
	it holds that
	\[
	\Pr\left[h(x)=h(y)\right]=\Pi_{i}\Pr\left[h_{i}(x)=h_{i}(y)\right]\ge p_{1}^{k}=2^{-\log\frac{1}{p_{1}}\cdot\frac{O(\log n)}{\log\frac{1}{p_{2}}}}=n^{-O(\rho)}~.
	\]
\end{proof}

\begin{lemma}
	\label{Lem:LSHPartitions}Let $\left(X,d_{X}\right)$ be a metric space, such that for every $r>0$, there is a $(r,t\cdot r,p_1,p_2)$-sensitive LSH family with parameter $\rho$. Then $\left(X,d_{X}\right)$ is $(t,n^{-O(\rho)})$-decomposable.
\end{lemma}
\begin{proof}
	
	Let $H'$ be an $\left(r,tr,n^{-O(\rho)},n^{-2}\right)$-sensitive LSH family, given by \claimref{claim:hash}.
	We will use $H'$ in order to construct a decomposition for $X$. Each hash function $h\in H'$ induces a partition $\mathcal{P}_h$, by clustering all points with the same hash value, i.e. $\mathcal{P}_h(x)=\mathcal{P}_h(y)\iff h(x)=h(y)$.
	However, in order to ensure that our partition will be $t\cdot r$-bounded, we modify it slightly.
	For $x\in X$, if there is a $y\in \mathcal{P}_h(x)$ with $d_X(x,y)>t\cdot r$, remove $x$ from $\mathcal{P}_h(x)$, and create a new cluster $\{x\}$. Denote by $\mathcal{P}'_h$ the resulting partition. $\mathcal{P}'_h$ is clearly $t\cdot r$-bounded, and we argue that every pair $x,y$ at distance at most $r$ is clustered together with probability at least $n^{-O(\rho)}$.
	Denote by $\chi_{x}$ (resp., $\chi_{y}$) the
	probability that $x$ (resp., $y$) was removed from $\mathcal{P}_h(x)$ (resp.,
	$\mathcal{P}_h(y)$). By the union bound on the at most $n$ points in $\mathcal{P}_h(x)$, we have that both $\chi_{x},\chi_y\le 1/n$.
	We conclude
	\[
	\Pr_{\mathcal{P}'_{h}}\left[\mathcal{P}'_{h}(x)=\mathcal{P}'_{h}(y)\right]\ge\Pr_{h\sim H}\left[h(x)=h(y)\right]-\Pr_{h}\left[\chi_{x}\vee\chi_{y}\right]\ge n^{-O(\rho)}-\frac{2}{n}=n^{-O(\rho)}~.
	\]
\end{proof}

Using \cite{AI06}, \lemmaref{Lem:LSHPartitions} implies that $\ell_{2}$ is
$(t,n^{-O(\sfrac{1}{t^{2}})})$-decomposable. Moreover, using \cite{Ngu13} for constant $p\in(1,2)$, \lemmaref{Lem:LSHPartitions} implies that $\ell_{p}$ is
$(t,n^{-O(\sfrac{\log^2 t}{t^{p}})})$-decomposable.
%

\section{Decomposition for $d$-Dimensional Euclidean Space}\label{app:dDimension}

In \sectionref{app:LSHtoDecomp}, using a reduction from LSH, we showed that $\ell_2$ is $(t,n^{-O(\sfrac{1}{t^{2}})})$-decomposable. Here, we will show that for dimension $d=o(\log n)$, using a direct approach, better decomposition could be constructed.

Denote by $B_d(x,r)$ the $d$ dimensional ball of radius $r$ around $x$ (w.r.t $\ell_2$ norm). $V_{d}(r)$ denotes the volume of  $B_d(x,r)$ (note that the center here is irrelevant).
Denote by $C_{d}(u,r)$ the volume of the intersection of two balls of
radius $r$, the centers of which are at distance $u$ (i.e. for $\|x-y\|_2=u$, $C_{d}(u,r)$ denotes the volume of $B_d(x,r)\cap B_d(y,r)$). We will use
the following lemma which was proved in \cite{AI06} (based on a lemma from \cite{FS02}).
%
\begin{lemma}
	\label{lem:ProbabilityCut}(\cite{AI06}) For any
	$d\ge2$ and $0\le u\le r$
	\[
	\Omega\left(\frac{1}{\sqrt{d}}\right)\cdot\left(1-\left(\frac{u}{r}\right)^{2}\right)^{\frac{d}{2}}\le\frac{C_{d}(u,r)}{V_{d}(r)}\le\left(1-\left(\frac{u}{r}\right)^{2}\right)^{\frac{d}{2}}~.
	\]
\end{lemma}
Using \lemmaref{lem:ProbabilityCut}, we can construct better decompositions:
\begin{lemma}
	For every $d\ge 2$ and $2\le t\le\sqrt{\sfrac{2d}{\ln d}}$, $\ell_2^d$ is $O(t,2^{-O(\frac{d}{t^2})})$-decomposable.
\end{lemma}
\begin{proof}
	Consider a set $X$ of $n$ points in  $\ell_2^d$, and fix $r>0$.
	Let $\mathcal{B}$ be some box which includes all of $X$ and such that each $x\in X$
	is at distance at least $t\cdot r$ from the boundary of $B$.
	We sample points $s_{1},s_{2}\dots$ uniformly at random from $\mathcal{B}$. Set $P_{i}=B_{X}(s_{i},\frac{t\cdot r}{2})\setminus\bigcup_{j=1}^{i-1}B_{X}\left(s_{j},\frac{t\cdot r}{2}\right)$. We sample points until $X=\bigcup_{i\ge1}P_i$.
	Then, the partition will be $\mathcal{P}=\left\{ P_{1},P_{2},\dots.\right\} $ (dropping empty clusters).
	
	It is straightforward that $\mathcal{P}$ is $t\cdot r$-bounded.
	Thus it will be enough to prove that every pair
	$x,y$ at distance at most $r$, has high enough probability to be
	clustered together.
	Let $s_{i}$ be the first point sampled in $B_{d}\left(x,\frac{t\cdot r}{2}\right)\cup B_{d}\left(y,\frac{t\cdot r}{2}\right)$.	
	By the minimality of $i$, $x,y\notin\bigcup_{j=1}^{i-1}B_{d}\left(s_{j},\frac{t\cdot r}{2}\right)$
	and thus both are yet un-clustered.
	If $s_{i}\in B_{2}\left(x,\frac{t\cdot r}{2}\right)\cap B_{2}\left(y\frac{t\cdot r}{2}\right)$
	then both $x,y$ join $P_{i}$ and thus clustered together. Using \lemmaref{lem:ProbabilityCut}
	we conclude,
	\begin{align*}
	\Pr_{\mathcal{P}}\left[\mathcal{P}(x)=\mathcal{P}(y)\right] & =\Pr\Bigg[s_{i}\in B_{2}\left(x,\frac{t\cdot r}{2}\right)\cap B_{2}\left(y,\frac{t\cdot r}{2}\right) \\
	& ~~~~~~~~~~~~~~~~~~~~~~~~~~~~~~~~~~~~ \Bigl\vert\hfil s_{i}\text{ is first in }B_{2}\left(x,\frac{t\cdot r}{2}\right)\cup B_{2}\left(y,\frac{t\cdot r}{2}\right)\Bigg]\\
	& \ge\frac{C_{d}(\|x-y\|_{2},\frac{t\cdot r}{2})}{2\cdot V_{d}(\frac{t\cdot r}{2})}\\
	& =\Omega\left(\frac{1}{\sqrt{d}}\right)\left(1-\left(\frac{\|x-y\|_{2}}{\frac{t\cdot r}{2}}\right)^{2}\right)^{\frac{d}{2}}\\
	& =\Omega\left(\frac{1}{\sqrt{d}}\right)\left(1-\frac{4}{t^{2}}\right)^{\frac{d}{2}}\\
	& =\Omega\left(e^{-\frac{2d}{t^{2}}-\frac{1}{2}\ln d}\right)=2^{-O\left(\sfrac{d}{t^{2}}\right)}~.
	\end{align*}
\end{proof}

%
%
%
%
%
%
%
%
%
%

\section{Acknowledgments}
We would like to thank an anonymous reviewer for useful comments.

\bibliographystyle{alpha}
\bibliography{HighSpanBib}

\end{document}